\documentclass[11pt]{article}

\usepackage{amsmath}
\usepackage{amssymb}
\usepackage{amsthm}
\usepackage{multirow}
\usepackage{color}

\newtheorem{theorem}{Theorem}
\newtheorem{corollary}[theorem]{Corollary}

\newtheorem{proposition}[theorem]{Proposition}

\newtheorem{definition}[theorem]{Definition}

\newtheorem{remark}[theorem]{Remark}

\newcommand{\bb}{\mathbb}
\newcommand{\Z}{\bb{Z}}

\newcommand{\cC}{\mathcal{C}}

\newcommand{\de}{\delta}

\newcommand{\Sig}{\Sigma}

\newcommand{\ta}{\theta}
\newcommand{\om}{\omega}
\newcommand{\Om}{\Omega}

\newcommand{\pr}{\prime}

\newcommand{\sm}{\setminus}

\newcommand{\lan}{\langle}
\newcommand{\ran}{\rangle}

\newcommand{\lf}{\lfloor}
\newcommand{\rf}{\rfloor}

\newcommand{\F}{\mathbb{F}}

\newcommand{\Fq}{\mathbb{F}_q}

\newcommand{\Fqm}{\mathbb{F}_{q^m}}

\newcommand{\Fqt}{\mathbb{F}_{q^2}}
\newcommand{\Fr}{\mathbb{F}_r}
\newcommand{\az}{a_0^{(i)}}
\newcommand{\ao}{a_1^{(i)}}
\newcommand{\ez}{e_0^{(i)}}
\newcommand{\eo}{e_1^{(i)}}
\newcommand{\dt}{d_2^{(i)}}
\newcommand{\doo}{d_1^{(i)}}

\newcommand{\bu}{{\bf u}}

\newcommand{\bv}{{\bf v}}
\newcommand{\bc}{{\bf c}}

\begin{document}

\title{Constructions of Maximum Distance Separable Symbol-Pair Codes Using Cyclic and Constacyclic Codes}

\author{Shuxing Li$^{\text{a}}$ and Gennian Ge$^{\text{b,c,}}$\thanks{Corresponding author. Email address: gnge@zju.edu.cn. Research  supported by the National Natural Science Foundation of China under Grant Nos.  11431003 and 61571310.}\\
  \footnotesize $^{\text{a}}$ School of Mathematical Sciences, Zhejiang University, Hangzhou 310027, Zhejiang, China\\
  \footnotesize $^{\text{b}}$ School of Mathematical Sciences, Capital Normal University, Beijing, 100048, China\\
\footnotesize $^{\text{c}}$ Beijing Center for Mathematics and Information Interdisciplinary Sciences, Beijing, 100048, China}
\date{}
\maketitle

\begin{abstract}
Symbol-pair code is a new coding framework which is proposed to correct errors in the symbol-pair read channel. In particular, maximum distance separable (MDS) symbol-pair codes are a kind of symbol-pair codes with the best possible error-correction capability. Employing cyclic and constacyclic codes, we construct three new classes of MDS symbol-pair codes with minimum pair-distance five or six. Moreover, we find a necessary and sufficient condition which ensures a class of cyclic codes to be MDS symbol-pair codes. This condition is related to certain property of a special kind of linear fractional transformations. A detailed analysis on these linear fractional transformations leads to an algorithm, which produces many MDS symbol-pair codes with minimum pair-distance seven.

\smallskip

\end{abstract}

\section{Introduction}

Motivated by high-density storage applications, a new coding framework named symbol-pair code was proposed in \cite{CB10,CB} to correct errors in the so-called symbol-pair read channel. Consider a scenario where we want to read data from certain storage medium. When the data is written in a very compact way and our data reader has relatively low resolution, instead of individual symbols, we can only receive overlapping pairs of symbols. Suppose the data symbols belong to an alphabet $\Sig$. Then, what we receive are pairs of symbols belonging to a different alphabet $\Sig \times \Sig$. In order to recover the original data reliably, we need a new coding scheme which is able to correct errors in this symbol-pair read channel.

Cassuto and Blaum laid the foundation of symbol-pair codes in \cite{CB10,CB}, which play the roles of error-correcting codes for the symbol-pair read channel. They presented several bounds and constructions, as well as a decoding algorithm for symbol-pair codes. The construction of symbol-pair codes are further studied in a series of papers, including algebraic constructions \cite{CL,CJKWY,KZL} and combinatorial constructions \cite{CJKWY}. Moreover, an efficient decoding algorithm of cyclic symbol-pair codes is proposed in \cite{YBS}.

In \cite{CJKWY}, the authors derived a Singleton-type bound for symbol-pair codes. Consequently, the concept of maximal distance separable (MDS) symbol-pair codes is proposed. The construction of MDS symbol-pair codes is interesting because they have the best possible capability against errors in the symbol-pair read channel. In general, there are two ways to construct MDS symbol-pair codes. The first one is direct construction using linear codes with appropriate properties, such as MDS codes \cite{CJKWY}, as well as cyclic and constacyclic codes \cite{KZL}. The second way is recursive construction employing the interleaving technique \cite{CJKWY,CKW}, the Eulerian graph \cite{CJKWY,CKW,KZL} and other combinatorial configurations \cite{CJKWY,CKW}.

In particular, we focus on the construction of $(n,d_p)_q$ MDS symbol-pair code whose minimum pair-distance $d_p$ is small. The known parameters of $(n,d_p)_q$ MDS symbol-pair codes with small $d_p$ are the following ones:
\begin{itemize}
\item[a)] $q \ge 2$, $n \ge 2$, $d_p \in \{2,3\}$ \cite{CJKWY},
\item[b)] $q \ge 2$, $n \ge 4$, $d_p=4$ \cite{CJKWY},
\item[c1)] $q$ is an even prime power, $n \le q+2$, $d_p=5$ \cite{CJKWY},
\item[c2)] $q$ is an odd prime, $5 \le n \le 2q+3$, $d_p=5$ \cite{CJKWY},
\item[c3)] $q$ is a prime power, $n\mid q^2-1$, $n>q+1$, $d_p=5$ \cite{KZL},
\item[c4)] $q$ is a prime power, $n=q^2+q+1$, $d_p=5$ \cite{KZL},
\item[c5)] $q\equiv 1 \pmod3$ is a prime power, $n=\frac{q^2+q+1}{3}$, $d_p=5$ \cite{KZL},
\item[d1)] $q$ is a prime power, $n=q^2+1$, $d_p=6$ \cite{KZL},
\item[d2)] $q$ is an odd prime power, $n=\frac{q^2+1}{2}$, $d_p=6$ \cite{KZL},
\item[e)] $q$ is an odd prime, $n=8$, $d_p=7$ \cite{CJKWY}.
\end{itemize}

In this paper, we follow the idea in \cite{KZL} to construct MDS symbol-pair codes by employing cyclic and constacyclic codes. We use $v_p(n)$ to denote the largest integer $a$, such that $p^a \mid n$, where $p$ is a prime. We obtain the following new classes of $(n,d_p)_q$ MDS symbol-pair codes with $d_p\in\{5,6\}$.
\begin{itemize}
\item[1)] Let $q$ be a prime power. Let $n$ and $r$ be two integers such that
$$
r \mid q-1, \; nr \mid q^3-1, \; nr \nmid q-1, \; (\frac{q-1}{r},n)=1.
$$
Then there exists an $(n,d_p)_q$ MDS symbol-pair code with $d_p=5$.

\item[2)] Let $q$ be a prime power, Let $n$ and $r$ be two integers such that
$$
nr \mid (q-1)(q^2+1), \; nr \nmid q^2-1, \; (\frac{q-1}{r},n)=1.
$$
Then there exists an $(n,d_p)_q$ MDS symbol-pair code with $d_p=6$.

\item[3)] Let $q$ be a prime power and $n \mid q^2-1$. If $n$ is odd or $n$ is even and $v_2(n)<v_2(q^2-1)$, then there exists an $(n,d_p)_q$ MDS symbol-pair code with $d_p=6$.
\end{itemize}
We remark that the class 1) (resp. class 2)) is an extension of the classes c4) and c5) (resp. classes d1) and d2)). More interestingly, for a class of cyclic codes, we find a necessary and sufficient condition which guarantees them to be MDS symbol-pair codes with minimum pair-distance $d_p=7$. We observe that this condition is related to the property of a special kind of linear fractional transformations. Moreover, we present a detailed analysis of these linear fractional transformations, which leads to a precise characterization of this condition. Using this characterization, we obtain many examples of MDS symbol-pair codes with minimum pair-distance $d_p=7$.

The rest of this paper is organized as follows. Section~\ref{sec2} gives a brief introduction to cyclic and constacyclic codes. Some preliminaries concerning symbol-pair codes and MDS symbol-pair codes are also presented. Employing cyclic and constacyclic codes, several constructions of MDS symbol-pair codes are presented in Section~\ref{sec3}. Section~\ref{sec4} concludes the paper.

\section{Preliminaries}\label{sec2}

\subsection{Cyclic and constacyclic codes}

Let $q$ be a prime power, $\Fq$ be a finite field and $\om \in \Fq^*$. An $\om$-constacyclic code $\cC$ is a linear code which is invariant under the constacyclic shift. Namely, if
$$
(c_0,c_1,\ldots,c_{n-1}) \in \cC,
$$
then
$$
(\om c_{n-1},c_0,\ldots,c_{n-2}) \in \cC.
$$
An $\om$-constacyclic code $\cC$ of length $n$ over $\Fq$ can be identified with an ideal of the principal ideal ring $\Fq[x]/(x^n-\om)$. Thus, $\cC$ can be generated by one element. There is a unique monic polynomial $g(x) \in \Fq[x]$ of minimum degree in $\cC$, such that $g(x) \mid x^n-\om$ and $\cC=\lan g(x)\ran$. This polynomial is called the {\it generator polynomial} of $\cC$. Given the ring $\Fq[x]/(x^n-\om)$ and a generator polynomial $g(x)$, an $\om$-constacyclic code $\cC=\lan g(x)\ran$ of length $n$ is determined, which is a linear subspace of $\Fq^n$ with dimension $n-\deg(g(x))$. When $\om=1$, an $\om$-constacyclic code is simply a cyclic code.

Suppose $\om \in \Fq^*$ is an element of order $r$ and $m$ is the smallest integer such that $nr \mid q^m-1$. Then we can find an element $\de \in \Fqm^*$ of order $nr$, such that $\om=\de^n$. Therefore the roots of $x^n-\om$ are of the form $\{ \de^{1+jr} \mid 0 \le j \le n-1\}$. Define $\Om=\{1+jr \mid 0 \le j \le n-1\}$. For $s \in \Om$, the $q$-cyclotomic coset modulo $nr$ containing $s$ is defined to be $C_s=\{q^is \pmod{nr} \mid 0 \le i \le m-1 \}$. Since $g(x) \in \Fq[x]$ and $g(x) \mid x^n-\om$, we have $g(x)=\prod_{s \in S}\prod_{j \in C_s} (x-\de^j)$, where $S \subset \Om$ is a subset of representatives of the $q$-cyclotomic cosets modulo $nr$.

For cyclic codes, we have the well-known BCH bound on the minimum distance. Similarly, we have the following BCH-type bound on the minimum distance of a constacyclic code, which is a slight generalization of \cite[Theorem 3]{KZL}.

\begin{proposition}\label{prop-BCH}
Let $q$ be a prime power and $n$ be a positive integer with $(n,q)=1$. Let $\om \in \Fq^*$ be an element of order $r$. Let $m$ be the smallest positive integer such that $nr \mid q^m-1$. Then there exists $\de \in \Fqm^*$, such that $\de$ has order $nr$ and $\om=\de^n$. Define $\xi=\de^r$. Let $\cC=\lan g(x)\ran \subset \Fq[x]/(x^n-\om)$ be an $\om$-constacyclic code with length $n$. Let $l$ be an integer with  $(l,n)=1$ and $d$ be an integer with $1 \le d \le n-1$. Suppose each element of $\{\de\xi^{li} \mid b\le i\le b+d-1\}$ is a root of the generator polynomial $g(x)$, where $b$ is an arbitrary integer. Then the minimum distance of $\cC$ is at least $d+1$.
\end{proposition}
\begin{proof}
The condition $(n,q)=1$ ensures that $g(x)$ has no repeated roots. Since each element belonging to $\{\de\xi^{li} \mid b\le i\le b+d-1\}$ is a root of $g(x)$, the matrix
$$
\begin{pmatrix}
1 & \de\xi^{bl} & \cdots & \de^{n-1}\xi^{(n-1)bl} \\
1 & \de\xi^{(b+1)l} & \cdots & \de^{n-1}\xi^{(n-1)(b+1)l} \\
\vdots & \vdots & & \vdots\\
1 & \de\xi^{(b+d-1)l} & \cdots & \de^{n-1}\xi^{(n-1)(b+d-1)l}
\end{pmatrix}
$$
is a submatrix of the parity matrix of $\cC$. Employing the condition $(l,n)=1$ and the property of the Vandermonde matrix, we conclude that any submatrix of the above one with $d$ columns must be nonsingular. Consequently, the minimum distance of $\cC$ is at least $d+1$.
\end{proof}

\subsection{Symbol-pair codes and MDS symbol-pair codes}

Let $\Sig$ be an alphabet consisting of $q$ elements. Given $\bu=(u_0,u_1,\ldots,u_{n-1})\in\Sig^n$, the {\it symbol-pair read vector} of $\bu$ is defined to be
$$
\pi(\bu)=((u_0,u_1),(u_1,u_2),\ldots,(u_{n-2},u_{n-1}),(u_{n-1},u_0)) \in (\Sig\times\Sig)^n.
$$
Let $\bu=(u_0,u_1,\ldots,u_{n-1})\in\Sig^n$ and $\bv=(v_0,v_1,\ldots,v_{n-1})\in\Sig^n$, the {\it pair-distance} between $\bu$ and $\bv$ is
$$
d_P(\bu,\bv)=|\{0\le i\le n-1 \mid (u_i,u_{i+1}) \ne (v_i,v_{i+1}) \}|,
$$
where the subscripts are regarded as integers modulo $n$. An $(n,M,d_p)_q$ symbol-pair code is a subset $\cC \subset \Sig^n$ with $|\cC|=M$, such that $d_p=\min\{d_P(\bu,\bv)\mid \bu,\bv\in\cC, \bu\ne\bv\}$. If $\Sig$ is a finite field $\Fq$, define the {\it pair-weight} of $\bu\in\Fq^n$ to be
$$
w_P(\bu)=|\{0\le i\le n-1 \mid (u_i,u_{i+1}) \ne (0,0)\}|,
$$
where the subscripts are regarded as integers modulo $n$. In particular, if the $(n,M,d_p)_q$ symbol-pair code $\cC$ is a linear subspace of $\Fq^n$, then $d_p=\min\{w_P(\bu)\mid \bu \ne (0,0,\ldots,0)\}$.

Let $\bu=(u_0,u_1,\ldots,u_{n-1})$ be the original vector. Let
$$
((u_0^{\pr},u_1^{\pr\pr}),(u_1^{\pr},u_2^{\pr\pr}),\ldots,(u_{n-2}^{\pr},u_{n-1}^{\pr\pr}),(u_{n-1}^{\pr},u_0^{\pr\pr})) \in (\Sig\times\Sig)^n
$$
be the received vector via the symbol-pair read channel. Then the number of {\it pair errors} is defined to be
$$
|\{0\le i\le n-1\mid (u_i,u_{i+1}) \ne (u_i^{\pr},u_{i+1}^{\pr\pr})\}|
$$
where the subscripts are regarded as integers modulo $n$. Similar to the classical error-correcting codes, an $(n,M,d_p)_q$ symbol-pair code can correct up to $\lf \frac{d_p-1}{2}\rf$ pair errors \cite[Proposition 3]{CB}. Hence, given $q$, $n$ and $M$, we aim to construct symbol-pair codes with $d_p$ as large as possible. To this end, we want to take advantage of the fruitful results concerning classical error-correcting codes. A first step is to understand the connection between symbol-pair codes and classical error-correcting codes.

The pair-distance was first introduced in \cite{CB10,CB}, which has been shown to be a well-defined metric. Recall that the Hamming distance between $\bu=(u_0,u_1,\ldots,u_{n-1})$ and $\bv=(v_0,v_1,\ldots,v_{n-1})$ is defined to be
$$
d_H(\bu,\bv)=|\{0 \le i \le n-1 \mid u_i\ne v_i\}|.
$$
In order to build a connection between the pair-distance and the Hamming distance, we need the following definition.

\begin{definition}
Let $S$ be a subset of $\{0,1,\ldots,n-1\}$. Thus, the elements of $S$ can be regarded as elements of $\Z_n$, the ring of integers modulo $n$. $S$ can be partitioned into a union of subsets, such that each subset consists of elements of $\Z_n$, which are consecutive in the sense of modulo $n$. Clearly, the partition of $S$ with smallest number of subsets is unique. Therefore, we define $L(S)$ to be the number of subsets in this unique partition.
\end{definition}

The following proposition reveals the connection between the pair-distance and the Hamming distance.

\begin{proposition}{\rm \cite[Proposition 1 and Theorem 2]{CB}}
Let $\bu=(u_0,u_1,\ldots,u_{n-1})$ and $\bv=(v_0,v_1,\ldots,v_{n-1})$ be two vectors of $\Sig^n$ with $0<d_H(\bu,\bv)<n$. Define $S=\{0 \le i\le n-1\mid u_i\ne v_i\}$. Then
$$
d_P(\bu,\bv)=d_H(\bu,\bv)+L(S).
$$
Therefore, we have $L(S)=d_P(\bu,\bv)-d_H(\bu,\bv) \le n-d_H(\bu,\bv)$. Together with $1 \le L(S) \le d_H(\bu,\bv)$, we have
$$
d_H(\bu,\bv)+1 \le d_P(\bu,\bv)\le \min\{2d_H(\bu,\bv),n\}.
$$
In addition,
$$
d_P(\bu,\bv)=\begin{cases}
  0 & \mbox{if $d_H(\bu,\bv)=0$}, \\
  n & \mbox{if $d_H(\bu,\bv)=n$}.
\end{cases}
$$
\end{proposition}

In particular, for linear symbol-pair codes, we have the following corollary concerning the relation between the Hamming weight and the pair-weight of a codeword.

\begin{corollary}\label{cor-pairweight}
Let $\cC$ be an $(n,M,d_p)_q$ symbol-pair code, which is a linear subspace of $\Fq^n$. For any $\bc=(c_0,c_1,\ldots,c_{n-1}) \in \cC$, define
$$
I(\bc)=L(\{0 \le i \le n-1 \mid c_i \ne 0\}).
$$
Suppose $0<w_H(\bc)<n$, where $w_H(\bc)$ denotes the Hamming weight of $\bc$. Then we have
\begin{equation}\label{eqn-pairweight}
w_P(\bc)=w_H(\bc)+I(\bc).
\end{equation}
Therefore, we have $I(\bc)=w_P(\bc)-w_H(\bc) \le n-w_H(\bc)$. Together with $1 \le I(\bc) \le w_H(\bc)$, we have
$$
w_H(\bc)+1\le w_P(\bc) \le \min\{2w_H(\bc),n\}.
$$
In particular, if the minimum Hamming distance of $\cC$ is $d<n$, then the minimum pair distance
\begin{equation}\label{ineqn-trivial}
d+1 \le d_p \le \min\{2d,n\}.
\end{equation}
\end{corollary}

Similar to classical error-correcting codes, there are several bounds providing fundamental restrictions on the parameters of symbol-pair codes. One of them is the following Singleton-type bound.

\begin{proposition}{\rm \cite[Theorem 2.1]{CJKWY}}
Let $q \ge 2$ and $2 \le d \le n$. If $\cC$ is an $(n,M,d_p)_q$ symbol-pair code, then $M\le q^{n-d_p+2}$.
\end{proposition}

The symbol-pair code $\cC$ achieving this Singleton-type bound is called an MDS symbol-pair code. We denote it by an $(n,d_p)_q$ MDS symbol-pair code. Below, we focus on the direct construction of MDS symbol-pair codes. In fact, classical MDS codes directly generate MDS symbol-pair codes.

\begin{proposition}{\rm \cite[Proposition 3.1]{CJKWY}}\label{prop-MDS}
If $\cC$ is an MDS code, then $\cC$ is an MDS symbol-pair code. Moreover, if $\cC$ is an $[n,n-d+1,d]_q$ MDS code with $d<n$, then $\cC$ is an $(n,d+1)_q$ MDS symbol-pair code.
\end{proposition}

Together with the knowledge concerning classical MDS codes, the above proposition implies that we have known a systematic construction for $(n,d_p)_q$ MDS symbol-pair codes with $q$ being a prime power and $2 \le d_p \le n \le q+1$. Below, we will focus on the construction of $(n,d_p)_q$ MDS symbol-pair codes with $q$ being a prime power and $n>q+1$.

We observe that if $\cC$ is a constacyclic code and is not MDS, then the lower bound in (\ref{ineqn-trivial}) can be improved.

\begin{proposition}\label{prop-consta}
Let $\cC$ be an $[n,k,d]_q$ constacyclic code with generator polynomial $g(x)$ and $d \le n-k$. Let $c(x) \in \cC$ be a codeword with Hamming weight $d^{\pr} \le n-k$. Then we have $I(c(x)) \ge 2$ and $w_P(c(x))\ge d^{\pr}+2$. In particular, $\cC$ is an $(n,q^k,d_p)_q$ symbol-pair code with $d_p \ge d+2$.
\end{proposition}
\begin{proof}
It suffices to show that $I(c(x))\ge 2$, which implies $w_P(c(x))\ge d^{\pr}+2$ by (\ref{eqn-pairweight}). Otherwise, we must have $I(c(x))=1$. This implies the indices of nonzero entries in $c(x)$ form one consecutive subset. Without loss of generality, we can assume that $c(x)=\sum_{i=0}^{d^{\pr}-1}c_ix^i$, where $c_i \in \Fq^*$ for each $0 \le i \le d^{\pr}-1$. Note that $g(x) \mid c(x)$. This leads to a contradiction since $\deg(g(x))=n-k \ge d^{\pr} > \deg(c(x))$. Therefore, we have $w_P(c(x))\ge d^{\pr}+2$. In particular, since $\cC$ is a linear code, the minimum pair-distance of $\cC$ equals its minimum nonzero pair-weight. Since $d\le n-k$, we can easily see that $d_p \ge d+2$ by Corollary \ref{cor-pairweight}.
\end{proof}

This proposition is an essential ingredient for the constructions in \cite{KZL} (see \cite[Lemma 5]{KZL}). In the following, we will employ cyclic and constacyclic codes to generate MDS symbol-pair codes.

\section{New constructions of MDS symbol-pair codes}\label{sec3}

Let $q$ be a prime power and $n$ be a positive integer. In this section, we are going to construct $(n,d_p)_q$ MDS symbol-pair codes with $d_p \in \{5,6,7\}$.

First, we consider the construction of MDS symbol-pair codes with $d_p=5$, which extends the results of {\rm \cite[Theorem 16]{KZL}} and {\rm \cite[Theorem 19]{KZL}}.

\begin{theorem}\label{thm-five}
Let $q$ be a prime power. Let $n$ and $r$ be two positive integers such that
$$
r \mid q-1, \; nr \mid q^3-1, \; nr \nmid q-1, \; (\frac{q-1}{r},n)=1.
$$
Then there exists an $(n,5)_q$ MDS symbol-pair code.
\end{theorem}
\begin{proof}
Let $\om \in \Fq^*$ be an element of order $r$. Let $\de \in \F_{q^3}^*$ be an element of order $nr$, such that $\de^n=\om$. Since $nr \nmid q-1$, we have $\de \in \F_{q^3}^* \setminus \Fq$, and the polynomial $g(x)=(x-\de)(x-\de^q)(x-\de^{q^2}) \in \Fq[x]$ divides $x^n-\om$. Let $\cC$ be the $\om$-constacyclic code $\lan g(x) \ran \subset\Fq[x]/(x^n-\om)$. Employing Proposition~\ref{prop-BCH} with $l=\frac{q-1}{r}$, we have the minimum distance of $\cC$ is at least three. In addition, by the Singleton bound, $\cC$ is an $[n,n-3,d]_q$ code with $3 \le d \le 4$. A direct application of Proposition~\ref{prop-MDS} and Proposition~\ref{prop-consta} shows that $\cC$ is an $(n,5)_q$ MDS symbol-pair code.
\end{proof}

\begin{remark}
By {\rm \cite[Corollary 7.4.4]{HP}}, when $n>2(q-1)$, the code $\cC$ in the above theorem must have minimum distance $3$. In addition, when $n=q^2+q+1$, $\cC$ is simply the Hamming code with minimum distance $3$. In this case, $\cC$ also achieves the pair-sphere packing bound {\rm \cite[Theorem 19]{CB}}.
\end{remark}

Next, we provide two constructions of MDS symbol-pair codes with $d_p=6$. The first one extends the results of {\rm \cite[Theorem 12]{KZL}} and {\rm \cite[Theorem 13]{KZL}}.

\begin{theorem}\label{thm-six-one}
Let $q$ be a prime power. Let $n$ and $r$ be two integers such that
$$
r \mid q-1, \; nr \mid (q-1)(q^2+1), \; nr \nmid q^2-1, \; (\frac{q-1}{r},n)=1.
$$
Then there exists an $(n,6)_q$ MDS symbol-pair code.
\end{theorem}
\begin{proof}
Let $\om \in \Fq^*$ be an element of order $r$. Let $\de \in \F_{q^4}^*$ be an element of order $nr$, such that $\de^n=\om$. Since $nr \nmid q^2-1$, we have $\de \in \F_{q^4}^* \setminus \Fqt$, and the polynomial $g(x)=(x-\de)(x-\de^q)(x-\de^{q^2})(x-\de^{q^3}) \in \Fq[x]$ divides $x^n-\om$. Let $\cC$ be the $\om$-constacyclic code $\lan g(x) \ran \subset\Fq[x]/(x^n-\om)$. Employing Proposition~\ref{prop-BCH} with $l=\frac{q-1}{r}$, we have that the minimum distance of $\cC$ is at least three. In addition, by the Singleton bound, $\cC$ is an $[n,n-4,d]_q$ code with $3 \le d \le 5$. Below, we are going to show that $d \ne 3$.

Assume the minimum distance of $\cC$ is three. Without loss of generality, we have a codeword $1+a_ix^i+a_jx^j$, where $1 \le i,j \le n-1$, $i \ne j$ and $a_i,a_j \in \Fq^*$. Thus, we have $1+a_i\de^i+a_j\de^j=0$. Since $nr \mid (q-1)(q^2+1)$, we get
$$
(1+a_i\de^i)^{(q-1)(q^2+1)}=(-a_j\de^j)^{(q-1)(q^2+1)}=1,
$$
which implies that $(1+a_i\de^i)^{q(q^2+1)}=(1+a_i\de^i)^{(q^2+1)}$. A direct computation leads to $\de^{qi}+\de^{q^3i}+a_i\de^{(q^3+q)i}=\de^i+\de^{q^2i}+a_i\de^{(q^2+1)i}$. Since $q^3+q \equiv q^2+1 \pmod{nr}$, we have $\de^{q^3+q}=\de^{q^2+1}$ and $\de^{q^3-1}=\de^{q^2-q}$. Consequently, we have $\de^{(q-1)i}+\de^{(q^3-1)i}=1+\de^{(q^2-1)i}$. Noting that $\de^{q^3-1}=\de^{q^2-q}$, we have $\de^{(q-1)i}+\de^{(q^2-q)i}=1+\de^{(q^2-1)i}$, which implies
$$
(\de^{(q-1)i}-1)(\de^{(q^2-q)i}-1)=0.
$$
This forces that $nr \mid (q-1)i$ for some $1 \le i \le n-1$. However, since $(\frac{q-1}{r},n)=1$, this is impossible.

Hence, the minimum distance of $\cC$ is either four or five. It is easily followed from Proposition~\ref{prop-MDS} and Proposition~\ref{prop-consta} that $\cC$ is an $(n,6)_q$ MDS symbol-pair code.
\end{proof}

When $n \mid q^2-1$, we have the following construction of $(n,6)_q$ MDS symbol-pair codes.

\begin{theorem}\label{thm-six-two}
Let $q$ be a prime power and $n$ be an integer with $n > q+1$ and $n \mid q^2-1$. Then
\begin{itemize}
\item[1)] There exists an $(n,6)_q$ MDS symbol-pair code when $n$ is odd.
\item[2)] There exists an $(\frac{n}{2},6)_q$ MDS symbol-pair code when $n$ is even.
\end{itemize}
\end{theorem}
\begin{proof}
1) Let $\de \in \Fqt^*\setminus\Fq$ be an element of order $n$ with $n$ being odd. The polynomial $g(x)=(x-\de^{-q})(x-\de^{-1})(x-\de)(x-\de^{q}) \in \Fq[x]$ divides $x^n-1$. Let $\cC_1$ be the cyclic code $\lan g(x) \ran \subset\Fq[x]/(x^n-1)$. Note that $\de^{-1}$ and $\de$ are two roots of $g(x)$ and $(2,n)=1$. Employing Proposition~\ref{prop-BCH} with $r=1$, $l=2$, $b=-1$ and $d=2$, we can see that $\de^{-1}$ and $\de$ are two consecutive roots and the minimum distance of $\cC_1$ is at least three. Together with the Singleton bound, $\cC_1$ is an $[n,n-4,d]_q$ code with $3 \le d \le 5$. When $4 \le d \le 5$, it is easily followed from Proposition~\ref{prop-MDS} and Proposition~\ref{prop-consta} that $\cC_1$ is an $(n,6)_q$ MDS symbol-pair code. When $d=3$, by Proposition~\ref{prop-MDS} and Proposition~\ref{prop-consta}, any codeword whose weight is greater than three has pair-weight at least six. Thus, by (\ref{eqn-pairweight}), it suffices to show that for each codeword $c(x) \in \cC$ with $w_H(c(x))=3$, we have $I(c(x)) \ge 3$. To this end, we are going to show that there is no codeword of the form $1+a_1x+a_ix^i$, where $2 \le i \le n-1$ and $a_1,a_i \in \Fq^*$. Below, we will split our discussion into two cases.

Firstly, assume there is a codeword $1+a_1x+a_2x^2$, where $a_1,a_2 \in \Fq^*$. Then we have the following system
$$
\begin{cases}
1+a_1\de+a_2\de^2=0, \\
1+a_1\de^{-1}+a_2\de^{-2}=0.
\end{cases}
$$
By solving this system, one can see that $a_1=-(\de+\frac{1}{\de})$. Therefore, we have $\de+\frac{1}{\de}\in \Fq^*$. Thus, $(\de+\frac{1}{\de})^q=\de+\frac{1}{\de}$, which implies that $(\de^{q+1}-1)(\de^{q-1}-1)=0$. Then, we have either $\de^{q+1}=1$ or $\de^{q-1}=1$. Namely, we have either $n \mid q+1$ or $n \mid q-1$. This is impossible because $n>q+1$.

Secondly, assume there is a codeword $1+a_1x+a_ix^i$, where $3 \le i \le n-2$ and $a_1,a_i \in \Fq^*$. Then we have the following system
$$
\begin{cases}
1+a_1\de+a_i\de^i=0, \\
1+a_1\de^{-1}+a_i\de^{-i}=0.
\end{cases}
$$
By solving the system, one can see that $a_1=-\frac{\de^{2i}-1}{\de^{2i-1}-\de}$ and $a_i=\frac{\de^{i+1}-\de^{i-1}}{\de^{2i-1}-\de}$. Therefore, we have $\frac{\de^{2i}-1}{\de^{2i-1}-\de}, \frac{\de^{i+1}-\de^{i-1}}{\de^{2i-1}-\de} \in \Fq^*$.
Since
$$
\frac{\de^{2i}-1}{\de^{2i-1}-\de}+\frac{\de^{i+1}-\de^{i-1}}{\de^{2i-1}-\de}=\frac{\de^{i+1}-1}{\de^i-\de}\in\Fq^*,
$$
and
$$
\frac{\de^{2i}-1}{\de^{2i-1}-\de}-\frac{\de^{i+1}-\de^{i-1}}{\de^{2i-1}-\de}=\frac{\de^{i+1}+1}{\de^i+\de}\in\Fq,
$$
we have
$$
\frac{\de^i-\de}{\de^{i+1}-1}+\frac{\de^{i+1}+1}{\de^i+\de}=\frac{(\de^{2i}-1)(\de^2+1)}{(\de^i+\de)(\de^{i+1}-1)} \in \Fq^*.
$$
Note that $\frac{\de^{2i}-1}{\de^{2i-1}-\de} \in \Fq^*$ and $\frac{\de^{i+1}-1}{\de^i-\de} \in \Fq^*$. Together with the above equation, we have
$$
\frac{(\de^{2i-1}-\de)(\de^2+1)}{(\de^i+\de)(\de^{i}-\de)}=\de+\frac{1}{\de} \in \Fq^*.
$$
However, as shown in the above, $\de+\frac{1}{\de} \in \Fq^*$ is impossible.

2) Let $\de \in \Fqt^*\setminus \Fq$ be an element of order $n$ with $n$ being even. Since $\de^{\frac{n}{2}}=-1$, the polynomial $g(x)=(x-\de^{-q})(x-\de^{-1})(x-\de)(x-\de^{q}) \in \Fq[x]$ divides $x^{\frac{n}{2}}+1$. Let $\cC_2$ be the (-1)-constacyclic code $\lan g(x) \ran \subset\Fq[x]/(x^{\frac{n}{2}}+1)$. Note that $\de^{-1}$ and $\de$ are two roots of $g(x)$. Employing Proposition~\ref{prop-BCH} with $r=2$, $l=1$, $b=-1$ and $d=2$, we can see that $\de^{-1}$ and $\de$ are two consecutive roots and the minimum distance of $\cC_2$ is at least three. Together with the Singleton bound, $\cC_2$ is an $[\frac{n}{2},\frac{n}{2}-4,d]_q$ code with $3 \le d \le 5$. The remaining part is similar to the proof of 1) and we omit it here.
\end{proof}

\begin{remark}
For $n \mid q^2-1$, $(n,6)_q$ MDS symbol-pair codes are constructed in Theorem~\ref{thm-six-two}, when $n$ is odd or $n$ is even and $v_2(n)<v_2(q^2-1)$. If $n$ is even and $v_2(n)=v_2(q^2-1)$, the construction in Theorem~\ref{thm-six-two} generates codes with minimum distance two, which are not MDS symbol-pair codes.
\end{remark}

\begin{remark}
By {\rm \cite[Corollary 7.4.4]{HP}}, the code $\cC_1$ (resp. $\cC_2$) in the above theorem has minimum distance $3\le d \le 4$ when $n>2(q-1)$ (resp. $n>4(q-1)$). Moreover, the codes $\cC_1$ and $\cC_2$ do have minimum distance $3$ in some cases. For instance, when $3 \mid n$, $\cC_1$ contains a codeword $1+x^{\frac{n}{3}}+x^{\frac{2n}{3}}$ with weight three and $\cC_2$ contains a codeword $1-x^{\frac{n}{6}}+x^{\frac{n}{3}}$ with weight three.
\end{remark}

In the following theorem, we will show that under certain condition, MDS symbol-pair codes with minimum pair-distance $d_p=7$ can be generated from certain cyclic codes.

\begin{theorem}\label{thm-seven}
Let $q$ be a prime power and $n$ be a positive integer with $n \mid q^2-1$ and $n>q+1$. Let $\de \in \Fqt^* \setminus \Fq$ be an element of order $n$. Let $\cC \subset \Fq[x]/(x^n-1)$ be an $[n,n-5,d]_q$ cyclic code having generator polynomial $g(x)=(x-\de^{-q})(x-\de^{-1})(x-1)(x-\de)(x-\de^q) \in \Fq[x]$. Then
\begin{itemize}
\item[1)] When $5 \le d \le 6$, $\cC$ is an $(n,7)_q$ MDS symbol-pair code.
\item[2)] When $d=4$ and $n$ is odd, $\cC$ is an $(n,7)_q$ MDS symbol-pair code if and only if for each $3 \le i \le n-3$, $\frac{\de^{i+1}-1}{\de^i-\de} \not\in \Fq^*$.
\end{itemize}
\end{theorem}
\begin{proof}
By the BCH bound and the Singleton bound, the minimum distance $4 \le d \le 6$. We only prove 2) since the proof of 1) is easy. When $d=4$, by Proposition~\ref{prop-MDS} and Proposition~\ref{prop-consta}, any codeword whose weight is greater than four has pair-weight at least seven. Thus, by (\ref{eqn-pairweight}), it suffices to show that for each codeword $c(x) \in \cC$ with $w_H(c(x))=4$, we have $I(c(x)) \ge 3$. Below, we are going to study the necessary and sufficient condition which ensures this restriction on codewords of weight four.

Suppose there is a codeword $c(x)$ of weight four, such that $I(c(x))=1$. Then without loss of generality, we can assume that $c(x)=1+a_1x+a_2x^2+a_3x^3$, where $a_1,a_2,a_3\in\Fq^*$. Consequently, the following system holds:
$$
\begin{cases}
  1+a_1+a_2+a_3=0, \\
  1+a_1\de+a_2\de^2+a_3\de^3=0,\\
  1+a_1\de^{-1}+a_2\de^{-2}+a_3\de^{-3}=0.
\end{cases}
$$
By solving this system, we have $a_2=1+\de+\frac{1}{\de}$. However, $\de+\frac{1}{\de} \in \Fq$ implies that $(\de^{q+1}-1)(\de^{q-1}-1)=0$. This leads to a contradiction since $n>q+1$.

Suppose there is a codeword $c(x)$ of weight four, such that $I(c(x))=2$. Then without loss of generality, we have the following two cases
\begin{itemize}
\item[i)] There is a codeword $c(x)=1+a_1x+a_2x^2+a_ix^i$, where $3 \le i \le n-2$ and $a_1,a_2,a_i\in\Fq^*$.
\item[ii)] There is a codeword $c(x)=1+a_1x+a_ix^i+a_{i+1}x^{i+1}$, where $3 \le i \le n-3$ and $a_1,a_i,a_{i+1}\in\Fq^*$.
\end{itemize}

For Case i), we must have the following system:
$$
\begin{cases}
  1+a_1+a_2+a_i=0, \\
  1+a_1\de+a_2\de^2+a_i\de^i=0,\\
  1+a_1\de^{-1}+a_2\de^{-2}+a_i\de^{-i}=0.
\end{cases}
$$
By solving this system, we have
$$
\frac{a_1}{a_2}=-\frac{\de^{i-2}-\de}{\de^{i-1}-1}-1, \quad a_2=\frac{\de^i-1}{\de^{i-1}-\de},
$$
which implies
$$
\frac{\de^{i-2}-\de}{\de^{i-1}-1} \in \Fq\sm\{-1\}, \quad \frac{\de^i-1}{\de^{i-1}-\de} \in \Fq^*.
$$
Thus,
\begin{align*}
\frac{\de^{i-1}-1}{\de^{i-2}-\de}-\frac{\de^i-1}{\de^{i-1}-\de}&=\frac{\de^{i-2}(\de+1)(\de-1)^2}{(\de^{i-1}-\de)(\de^{i-2}-\de)}\in \Fq^*,\\
\frac{\de^{i-1}-\de}{\de^i-1}-\frac{\de^{i-2}-\de}{\de^{i-1}-1}&=\frac{\de^{i-1}(\de-1)^2}{(\de^i-1)(\de^{i-1}-1)}\in\Fq^*.
\end{align*}
By comparing the right hand side of the above two equations, we have $1+\frac{1}{\de} \in \Fq^*$, which is impossible.

For Case ii), we must have the following system:
$$
\begin{cases}
  1+a_1+a_i+a_{i+1}=0, \\
  1+a_1\de+a_i\de^i+a_{i+1}\de^{i+1}=0,\\
  1+a_1\de^{-1}+a_i\de^{-i}+a_{i+1}\de^{-(i+1)}=0.
\end{cases}
$$
If $n$ is even, the above system holds if $i=\frac{n}{2}$, $a_1=a_{\frac{n}{2}+1}=-1$ and $a_{\frac{n}{2}}=1$. Hence, the condition of $n$ being odd is necessary. By solving the above system, we have
$$
a_1=-\frac{\de^{i+1}-1}{\de^i-\de},\quad a_i=\frac{\de^{i+1}-1}{\de^i-\de}, \quad a_{i+1}=-1.
$$
Thus, the above system does not hold, if and only if for each $3 \le i \le n-3$, $\frac{\de^{i+1}-1}{\de^i-\de} \not\in \Fq^*$. Therefore, we complete the proof.
\end{proof}

Given an integer $3 \le i \le n-3$, $\frac{\de^{i+1}-1}{\de^i-\de}=\ta \in \Fq^*$ is equivalent to $\de^i=\frac{1-\ta\de}{-\ta+\de}$ for $\ta\in\Fq^*$. Thus, the necessary and sufficient condition in 2) of Theorem~\ref{thm-seven} is related to the property of the linear fractional transformation $\frac{1-\ta\de}{-\ta+\de}$ with respect to $\de$, where $\ta \in \Fq^*$. This provides a motivation to study this special type of linear fractional transformation. Using the result derived in the Appendix, we have the following theorem which gives a more precise characterization of the necessary and sufficient condition.

\begin{theorem}\label{thm-seven-adv}
Let $q$ be a prime power and $n$ be an integer with $n \mid q^2-1$ and $n>q+1$. Let $\de \in \Fqt \setminus \Fq$ be an element of order $n$. Let $x^2-bx-c$ be the monic minimal polynomial of $\de$ over $\Fq$. For an integer $i \ge 2$, define
\begin{equation}\label{aa}
\az=\sum_{j=0}^{\lf \frac{i-2}{2} \rf} {i-2-j \choose j}b^{i-2-2j}c^{j+1}, \quad \ao=\sum_{j=0}^{\lf \frac{i-1}{2} \rf} {i-1-j \choose j}b^{i-1-2j}c^{j}.
\end{equation}
Let $\cC \subset \Fq[x]/(x^n-1)$ be an $[n,n-5,d]_q$ cyclic code having generator polynomial $g(x)=(x-\de^{-q})(x-\de^{-1})(x-1)(x-\de)(x-\de^q)$. Then $\cC$ is an $[n,n-5,d]_q$ code with $4 \le d \le 6$. When $5\le d\le 6$, $\cC$ is an $(n,7)_q$ MDS symbol-pair code. When $d=4$ and $n$ is odd, $\cC$ is an $(n,7)_q$ MDS symbol-pair code if and only if for each $3 \le i \le n-3$, one of the following holds:
\begin{itemize}
\item[1)] $\ao=0$,
\end{itemize}
or when $\ao \ne 0$,
\begin{itemize}
\item[2)] if $\ao=1$, then $\az\ne-b$ or $c=1$,
\item[3)] if $\az=0$, then $\ao\ne\frac{1}{c}$ or $b=0$,
\item[4)] if $\az\ne0$ and $\ao\ne1$, then $\ao c=1$ or $\frac{\ao b+\az}{\ao-1}\ne\frac{\ao c-1}{\az}$.
\end{itemize}
\end{theorem}
\begin{proof}
The conclusion is a direct application of Theorem~\ref{thm-seven} and Corollary~\ref{cor-coeff}.
\end{proof}

\begin{remark}\label{rem}
By the sphere packing bound, when $n(n-1) \ge \frac{2q^5}{(q-1)^2}$, the code $\cC$ in the above theorem has minimum distance $d=4$.
\end{remark}

The above theorem and remark suggest an algorithm which aim to construct $(n,7)_q$ MDS symbol-pair codes with $n \mid q^2-1$, $n(n-1)\ge \frac{2q^5}{(q-1)^2}$ and $n$ being odd. We run a numerical experiment for all pairs
$$
\{(q,n)\mid \mbox{$q$ prime power}, q \le 100, n \mid q^2-1, \mbox{$n$ odd}, n >q+1 \}.
$$
For these instances, the corresponding $[n,n-5,d]_q$ code $\cC$ in Theorem~\ref{thm-seven-adv} always has $d=4$. The code $\cC$ is an $(n,7)_q$ MDS symbol-pair code whenever $q$ is odd, except for $(q,n)\in \{(59,435),(67,561),(83,861)\}$. Moreover, the experimental result suggests that $\cC$ is not an MDS symbol-pair code when $q$ is even. However, it seems not easy to prove that $q$ being odd is a necessary condition for $\cC$ being an $(n,7)_q$ MDS symbol-pair code.

\section{Conclusion}\label{sec4}

Following the idea in \cite{KZL}, we use cyclic and constacyclic codes to construct MDS symbol-pair codes with minimum pair-distance $d_p \in \{5,6,7\}$ in this paper. Our constructions extend the results in \cite{KZL}. Moreover, we derive a necessary and sufficient condition which ensures a class of cyclic code to be MDS symbol-pair codes. This condition is related to the property of a special kind of linear fractional transformations. We study these linear fractional transformations in detail and propose a more precise characterization of the necessary and sufficient condition. This characterization leads to an algorithm aiming to construct MDS symbol-pair codes with minimum pair-distance $d_p=7$. We believe that a deeper understanding on this characterization may bring new classes of MDS symbol-pair codes.

We observe that most of the known constructions of $(n,d_p)_q$ MDS symbol-pair codes focus on the case where $d_p$ is small. In this case, if we use an $[n,k,d]_q$ linear code to construct a symbol-pair code, then the difference $d_p-d$ is necessarily small. Thus, it is relatively easy to show that the required minimum pair-distance is achieved. It is an interesting research problem to consider the constructions of MDS symbol-pair codes with large minimum pair-distances.

\section*{Appendix}

Let $q$ be a prime power. For $u,v,w,z \in \Fq$ and $\de \in \Fqt$, define a linear fractional transformation from $\Fqt$ to $\Fqt$ by
$$
f_{u,v,w,z}(\de)=\frac{u+v\de}{w+z\de},
$$
where $w+z\de \ne 0$ and $uz-vw \ne 0$. We further assume that $z \ne 0$, since otherwise, $f_{u,v,w,z}$ degenerates into a linear function. Below, we will study this special kind of linear fractional transformation. In particular, suppose $\de \in \Fqt \setminus \Fq$, we will present a necessary and sufficient condition such that
$$
\de^i=\frac{u+v\de}{w+z\de}
$$
for some integer $i$. This condition provides a criterion to determine whether the linear fractional transformation $f_{u,v,w,z}$ maps $\de$ to an element belonging to the multiplicative cyclic group generated by $\de$.

\begin{proposition}\label{prop-coeff}
Let $\de \in \Fqt \setminus \Fq$. Let $x^2-bx-c$ be the monic minimal polynomial of $\de$ over $\Fq$. For an integer $i \ge 2$, define
\begin{equation}\label{a}
\az=\sum_{j=0}^{\lf \frac{i-2}{2} \rf} {i-2-j \choose j}b^{i-2-2j}c^{j+1}, \quad \ao=\sum_{j=0}^{\lf \frac{i-1}{2} \rf} {i-1-j \choose j}b^{i-1-2j}c^{j}.
\end{equation}
Then for $i \ge 0$, $\de^i=\frac{u+v\de}{w+z\de}$ if and only if one of the following holds:
\begin{itemize}
\item[1)] If $i=0$, then $u=w$, $v=z$.
\item[2)] If $i=1$, then $b=\frac{v-w}{z}$ and $c=\frac{u}{z}$.
\item[3)] If $i \ge 2$, then
$$
\ao \ne 0, \quad b=-\frac{\az}{\ao}+\frac{v}{z\ao}-\frac{w}{z}, \quad c=-\frac{w\az}{z\ao}+\frac{u}{z\ao}.
$$
\end{itemize}
\end{proposition}
\begin{proof}
1) and 2) are trivial. We only consider 3) below. Since $\de^i=\frac{u+v\de}{w+z\de}$ and $z \ne 0$, we have $\de^{i+1}+\frac{w}{z}\de^{i}-\frac{v}{z}\de-\frac{u}{z}=0$. Therefore, $\de$ is a root of the polynomial $x^{i+1}+\frac{w}{z}x^{i}-\frac{v}{z}x-\frac{u}{z}$ and
$$
x^{i+1}+\frac{w}{z}x^{i}-\frac{v}{z}x-\frac{u}{z} \equiv 0 \pmod{x^2-bx-c}.
$$
For an integer $i\ge0$, we define a polynomial $T_i(x)=x^{i+1}+\frac{w}{z}x^{i}$. For any $i \ge 2$, we have the following recurrence relation:
\begin{align*}
T_i(x)&\equiv x^{i+1}+\frac{w}{z}x^{i} \\
      &\equiv bx^{i}+cx^{i-1}+\frac{w}{z}(bx^{i-1}+cx^{i-2}) \\
      &\equiv b(x^i+\frac{w}{z}x^{i-1})+c(x^{i-1}+\frac{w}{z}x^{i-2}) \\
      &\equiv bT_{i-1}(x)+cT_{i-2}(x) \pmod{x^2-bx-c}.
\end{align*}
By employing this recurrence relation repeatedly, we have
\begin{align*}
T_i(x)&\equiv \dt T_2(x)+\doo T_1(x) \\
      &\equiv \eo T_1(x)+\ez T_0(x) \pmod{x^2-bx-c},
\end{align*}
where $\doo,\dt,\ez,\eo \in \Fq$. Now, we aim to determine $\ez$ and $\eo$ explicitly. The recurrence relation implies that $T_0(x)$ necessarily originates from $T_2(x)$ by subtracting a proper multiple of $x^2-bx-c$. Since $T_2(x)\equiv bT_{1}(x)+cT_{0}(x) \pmod{x^2-bx-c}$, we have $\ez=c \dt$. Apparently, $\dt$ is a summation of monomials regarding of $b$ and $c$. More precisely, suppose $i-2$ can be expressed as an ordered sum containing $i-2-2j$ ones and $j$ twos. Then this ordered sum corresponds to a monomial $b^{i-2-2j}c^j$ in the summation of $\dt$. Recall that there are $\binom{i-2-j}{j}$ ways to decompose $i-2$ into distinct ordered sums containing $i-2-2j$ ones and $j$ twos. Therefore, we have
$$
\dt=\sum_{j=0}^{\lf \frac{i-2}{2} \rf} {i-2-j \choose j}b^{i-2-2j}c^{j},
$$
and
$$
\ez=c\dt=\sum_{j=0}^{\lf \frac{i-2}{2} \rf} {i-2-j \choose j}b^{i-2-2j}c^{j+1}=\az.
$$
Similarly, by analyzing the decomposition of $i-1$ into ordered sums consisting of ones and twos, we have
$$
\eo=\sum_{j=0}^{\lf \frac{i-1}{2} \rf} {i-1-j \choose j}b^{i-1-2j}c^{j}=\ao.
$$
Consequently,
\begin{align*}
x^{i+1}+\frac{w}{z}x^{i}-\frac{v}{z}x-\frac{u}{z} &\equiv T_i(x)-\frac{v}{z}x-\frac{u}{z} \\
                                                  &\equiv \ao T_1(x)+\az T_0(x)-\frac{v}{z}x-\frac{u}{z} \\
                                                  &\equiv \ao x^2+(\az+\frac{w\ao}{z}-\frac{v}{z})x+\frac{w\az}{z}-\frac{u}{z} \\ &\equiv 0 \pmod{x^2-bx-c}.
\end{align*}
Hence, we must have $\ao \ne 0$ and $x^2+(\frac{\az}{\ao}+\frac{w}{z}-\frac{v}{z\ao})x+\frac{w\az}{z\ao}-\frac{u}{z\ao}=x^2-bx-c$. The conclusion follows by comparing the coefficients.
\end{proof}

Particularly, given $\de \in \Fqt \setminus \Fq$ and an integer $i\ge 2$, we have the following easy criterion to determine if $\de^i=\frac{1-\ta\de}{-\ta+\de}$ for some $\ta \in \Fq^*$.

\begin{corollary}\label{cor-coeff}
Let $\de \in \Fqt \setminus \Fq$. Let $x^2-bx-c$ be the monic minimal polynomial of $\de$ over $\Fq$. For an integer $i \ge 2$, $\de^i=\frac{1-\ta\de}{-\ta+\de}$ for some $\ta \in \Fq^*$ if and only if $\ao \ne 0$ and one of the following condition holds
\begin{itemize}
\item[1)] If $\ao=1$, then $\az=-b$ and $c \ne 1$,
\item[2)] If $\az=0$, then  $\ao=\frac{1}{c}$ and $b \ne 0$,
\item[3)] If $\az \ne 0$ and $\ao \ne 1$, then $\ao c\ne 1$ and $\frac{\ao b+\az}{\ao-1}=\frac{\ao c-1}{\az}$.
\end{itemize}
where $\az$ and $\ao$ are defined in {\rm (\ref{a})}. Moreover, let $\Fr$ be a subfield of $\Fq$. If $b,c \in \Fr$, then $\de^i=\frac{1-\ta\de}{-\ta+\de}$ for some $i\ge2$ only if $\ta \in \Fr$.
\end{corollary}
\begin{proof}
By setting $u=z=-1$ and $v=w=\ta$ in Proposition~\ref{prop-coeff}, we have $\de^i=\frac{1-\ta\de}{-\ta+\de}$ for some $\ta \in \Fq^*$ if and only if
$$
b=\frac{(\ao-1)\ta-\az}{\ao}, \quad c=\frac{\az\ta+1}{\ao}.
$$
If $\az=0$ and $\ao=1$, then we have $b=0$ and $c=1$, which is impossible since $x^2-1$ is reducible over $\Fq$. If either $\ao=1$ or $\az=0$, then the Condition 1) or the Condition 2) holds. If $\az \ne 0$ and $\ao \ne 1$, the Condition 3) is derived from the expressions of $b$ and $c$. Suppose $b$ and $c$ belong to a subfield $\Fr$, then $\az,\ao \in \Fr$ by definition. Since we have either $\az\ne0$ or $\ao\ne1$, it is easy to see that $\ta\in\Fr$.
\end{proof}

\subsection*{Acknowledgements}

The authors wish to thank the anonymous reviewers for their comments which are very helpful to improve the paper. The first author would like to express his gratitude to Prof. Maosheng Xiong, Hong Kong University of Science and Technology, for the enlightening discussions on this topic.

\end{document}